\documentclass[conference,10pt]{IEEEtran}
\usepackage{amsfonts,amsmath,mathrsfs,}
\usepackage{amssymb}
\usepackage{graphicx}
\usepackage{algorithm}
\usepackage{algorithmic}

\usepackage{color}
\usepackage{euscript}

\newcommand{\Co}{{\mathscr{C}}}
\newcommand{\QCo}{{\EuScript{Q}}}
\newcommand{\Hm}{{\mathbf{H}}}
\newcommand{\HH}{{\mathbf{H}}}

\newcommand{\hx}{{\Hm_X}}
\newcommand{\hz}{{\Hm_Z}}
\newcommand{\hxi}{{\hat{\Hm}_X}}
\newcommand{\hzi}{{\hat{\Hm}_Z}}
\newcommand{\hxq}{{\Hm_X^{(q)}}}
\newcommand{\hzq}{{\Hm_Z^{(q)}}}
\newcommand{\x}[2]{x_{#1,#2}}
\newcommand{\z}[2]{z_{#1,#2}}
\newcommand{\xq}[2]{{x^{(q)}_{#1,#2}}}
\newcommand{\zq}[2]{{z^{(q)}_{#1,#2}}}
\newcommand{\cxq}{{\Co_X^{(q)}}}
\newcommand{\czq}{{\Co_Z^{(q)}}}
\newcommand{\CoX}{{\Co_X}}
\newcommand{\CoZ}{{\Co_Z}}
\newcommand{\CoXi}{\hat{\CoX}}
\newcommand{\CoZi}{\hat{\CoZ}}
\newcommand{\CoXq}{{\cxq}}
\newcommand{\CoZq}{{\czq}}
\newcommand{\dX}{{d_X}}
\newcommand{\dZ}{{d_Z}}
\newcommand{\dQ}{{d_{\QCo}}}
\newcommand{\kQ}{{k_{\QCo}}}

\newcommand{\F}{\mathbb{F}}
\newcommand{\Fq}{{\mathbb{F}_q}}
\newcommand{\fm}{{\mathbb{F}_{2^m}}}
\newcommand{\ft}{{\mathbb{F}_2}}
\newcommand{\ff}{{\mathbb{F}_4}}
\newcommand{\M}[2]{{\cal M}_{#1^#2}}

\newcommand{\GX}{{\cal G}_X}
\newcommand{\GZ}{{\cal G}_Z}

\newcommand{\SCycX}[1]{{\text{Cyc}_X(#1)}}

\newcommand{\Cyc}{{\cal C}}
\newcommand{\Edge}{{\cal E}}

\newcommand{\eqdef}{\stackrel{\text{def}}{=}}

\newtheorem{defn}{Definition}
\newtheorem{definition}{Definition}
\newtheorem{rmk}{Remark}
\newtheorem{lemma}{Lemma}
\newtheorem{thm}{Theorem}

\newtheorem{proposition}{Proposition}

\begin{document}
\title{New constructions of CSS codes obtained by moving to higher alphabets}

\author{\IEEEauthorblockN{Iryna Andriyanova \IEEEauthorrefmark{1}, Denise Maurice\IEEEauthorrefmark{2}, Jean-Pierre Tillich\IEEEauthorrefmark{2}}
\IEEEauthorblockA{
\IEEEauthorrefmark{1} ETIS group, ENSEA/UCP/CNRS-UMR8051, France\\
\IEEEauthorrefmark{2}  INRIA, Equipe Secret, Domaine de Voluceau BP 105, F-78153 Le Chesnay cedex, France.}}
\maketitle

\begin{abstract} We generalize a construction of non-binary quantum LDPC codes over $\F_{2^m}$ due to \cite{KHIS11a} and apply it in particular to toric codes. We obtain
in this way not only codes with better rates than toric codes  but also improve dramatically the performance of standard iterative decoding. Moreover,
the new codes obtained in this fashion inherit the distance properties of the underlying toric codes and have therefore a minimum distance which 
grows as the square root of the length of the code for fixed $m$.
\end{abstract}

\section{Introduction}

LDPC codes~\cite{Gal63a} and their variants are one of the most
 satisfying answers to the problem of devising codes
 guaranteed by Shannon's theorem. 
 They display outstanding performance for a large class of error models with a 
fast decoding algorithm. Generalizing these codes to the quantum setting seems a promising way to 
devise powerful quantum error correcting codes for protecting,
 for instance, the 
very fragile superpositions manipulated in a quantum computer.
It should be emphasized that a fast decoding algorithm could be even more crucial in the quantum setting than in the classical one. In the classical case, when  error correction codes are used for communication over a noisy channel, the decoding time translates directly into communication delays. This has been the driving motivation to devise decoding schemes of low complexity, and is likely to be important in the quantum setting as well. However, there is an important additional motivation for efficient decoding in the quantum setting. Quantum computation is likely to require active stabilization. The decoding time thus translates into computation delays, and most importantly in error suppression delays. If errors accumulate faster than they can be identified, quantum computation may well become infeasible: fast decoding is an essential ingredient to fault-tolerant computation.

Quantum 
generalizations of LDPC codes have indeed been proposed in
\cite{MMM04a}. 
However, it has turned out that the design of high performance quantum
LDPC codes is much more complicated than in the classical setting. 
This is due to several reasons, the most obvious of which being that
the parity-check matrix of quantum LDPC codes must satisfy certain 
orthogonality constraints. 
This complicates significantly the construction of such codes.  In
particular, the plain random constructions that work so well in the
classical setting are pointless here. There have been a number of
attempts at overcoming this difficulty and a variety of methods for
constructing quantum LDPC codes have been proposed
\cite{Pos01a,Kit03a,MMM04a,COT05a,COT07a,LG06a,LG08a,HH07a,IofMez07a,Djo08a,SRK08a,Aly07b,Aly08a,HBD08a,TZ09a,TL10a,KHIS11a}. However,
with the exception of  \cite{TZ09a} which gives a construction of LDPC codes with minimum distance of the order of
the square root of the blocklength, all of these constructions suffer from disappointingly small
minimum distances, namely whenever they have non-vanishing rate and
parity-check matrices with bounded row-weight, their minimum distance is
either proved to be {\em bounded}, or unknown and with little hope
for unboundedness. 

The point has been made several times that minimum distance
is not everything, because there are complex decoding issues involved,
whose behavior depends only in part on the minimum distance, and also
because a poor asymptotic behavior may be acceptable when one limits
oneself to practical lengths. This is illustrated for instance in our case by the codes constructed in  \cite{KHIS11a} whose 
performance under iterative decoding is quite good even if their minimum distance might be bounded. 
Their construction can be summarized as follows. There are three ingredients:
\begin{itemize}
\item[(i)]
The starting point is a CSS quantum code associated to a couple
$(\CoX,\CoZ)$ of  binary LDPC codes satisfying $\CoZ^\perp \subset \CoX$ (see Section \ref{sec:CSS}) obtained from a construction due to
\cite{HH07a}. These LDPC codes have parity check matrices $\hx$ and $\hz$ which are $(2,L)$-regular, meaning that each column contains exactly 
2 ``1'''s and each row contains exactly $L$-ones. 
\item[(ii)]
 From this construction,  a pair of $q$-ary LDPC codes 
$(\CoXq,\CoZq)$ is deduced which satisfies $\CoZq^\perp \subset \CoXq$, where $q$ is some power of two, $q=2^m$. These codes
have parity-check matrices $\hxq$ and $\hzq$ of the same size as $\hx$ and $\hx$ respectively and which have nonzero entries 
whenever the corresponding entry of $\hx$ (respectively $\hz$) is equal to $1$, that is
\begin{equation}
\label{eq:condition1}
\xq{i}{j} \neq 0 \Leftrightarrow \x{i}{j} = 1,\;\;\zq{i}{j} \neq 0 \Leftrightarrow \z{i}{j} = 1,
\end{equation}
where $\xq{i,j},\x{i,j},\zq{i}{j},\z{i}{j}$ denote the entry corresponding to the $i$-th row and the $j$-th column of
$\hxq,\hx,\hzq,\hz$ respectively.
\item[(iii)]
By denoting the length of $\CoX,\CoZ,\CoXq,\CoZq$ by $n$, and by replacing each entry of $\hxq$ and $\hzq$ in the finite field $\Fq$  over 
$q$ elements
 by a binary matrix of size $2^m \times 2^m$, through a ring isomorphism
$A : \fm \rightarrow \M{2}{m} $ where $\M{2}{m}$ is a certain subring of binary $2^m \times 2^m$ matrices (i.e. a one-to-one mapping preserving field addition and
multiplication), a pair of two parity-check matrices $(\hxi,\hzi)$ is obtained. They define a pair of binary codes $(\CoXi,\CoZi)$ of length $n \times m$
satisfying the CSS condition $\CoZi^\perp \subset \CoXi$. 
\end{itemize}
The point of this construction is that the new quantum code associated to the pair $(\CoXi,\CoZi)$ can now be decoded
on the extension field $\Fq$ and this improves dramatically the performance in the same way as the performance of classical 
binary $(2,L)$ regular LDPC codes is improved by moving to a larger extension field $\fm$ as shown in \cite{Hu02,HEA05}.

Our purpose in this article is here to generalize the construction of \cite{KHIS11a} and to show that it can be applied to any 
pair of binary codes $(\CoX,\CoZ)$ satisfying $\CoZ^\perp \subset \CoX$ which are LDPC codes which have parity check matrices
which have exactly $2$ ``1'''s per column\footnote{In other words they are cycle codes  of a graph \cite{HakBre68}.}, not only the particular family of 
quasi-cyclic codes of this type which are 
constructed in \cite{HH07a}. We apply this generalized construction to  the toric codes of \cite{Kit03a} which are a particular instance of the CSS construction corresponding
to a pair of LDPC codes $(\CoX,\CoZ)$ which are $(2,4)$-regular.
 It presents the advantage of having 
a minimum distance which grows like the square root of the length but has also the drawback to be able to encode only $2$ qubits. We obtain in this
way a new code family  which displays several attractive features compared to the toric code family:\\
(i) it has the same two dimensional structure as toric codes, this might turn out to  very helpful for its implementation.
It represents for instance a quite  attractive code choice for performing quantum fault-tolerant computation \cite{Kit03a}.\\
(ii) it inherits the distance properties from  the underlying toric code and has therefore a minimum distance which 
grows like the square root of the length,\\
(iii) the number of encoded qubits is not constant anymore as for toric codes but grows as $2m$ where
$m$ is the degree of the extension field,\\
(iv) whereas iterative decoding displays very bad performances when applied to toric codes, plain iterative decoding behaves much 
better for this new family of codes and when  $m=9$ for instance, we obtain codes for which iterative decoding performs quite well
(see Section \ref{sec:results}).

Apart from the practical relevance of the codes constructed, there is also a theoretical aspect.  This shows for instance that it is possible to obtain families of CSS codes with a prescribed degree distribution on the check nodes
 with an unbounded minimum distance with the construction strategy of \cite{KHIS11a}. It is questionable whether or not
the codes constructed in \cite{KHIS11a} meet this property (one of the drawback of the codes constructed there is that they start
with a certain construction of quasi-cyclic CSS codes which can be easily proved to have bounded minimum distance).

\section{CSS codes and Tanner graphs}\label{sec:CSS}

\paragraph{{\bf CSS codes}}
The codes constructed in this paper fall into the category of Calderbank-Shor-Steane (CSS) codes \cite{CS96a,Ste96b}
which belong to a more general class of quantum codes called stabilizer codes \cite{Got97a,CRSS98a}.
The first class is described with the help of a pair of mutually orthogonal binary codes, whereas
the second class is given by an additive self-orthogonal code over $\ff$ with respect to the trace hermitian product.
Quantum codes on $n$ qubits are  linear subspaces of a Hilbert space of dimension $2^n$ and
do not necessarily have a compact representation in general. The nice feature of stabilizer codes is that they allow 
to define
such a space with the help of a very short representation, which is given here by
 a set of generators of the aforementioned additive code. 
Each generator is viewed as an element of the Pauli group on $n$ qubits and the quantum code
is then nothing but the space stabilized by these Pauli group elements. Moreover, the set of errors that such a quantum code
can correct can also be deduced directly from this discrete representation. For the subclass of CSS 
codes, this representation in terms of additive self-orthogonal codes is equivalent to a representation in terms
of a pair $(\CoX,\CoZ)$ of binary linear codes satisfying the condition
$\CoZ^\perp \subset \CoX$. The {\em quantum minimum distance} of such a CSS code is given by
\begin{eqnarray}
\label{eq:distance}
\dQ & \eqdef &\min \{ \dX, \dZ\}, \;\; \text{where}\\
\dX & \eqdef & \min \{|x|,x \in \CoX \setminus \CoZ^\perp\}, \nonumber\\
\dZ &\eqdef & \min \{|x|,x \in \CoZ \setminus \CoX^\perp\}. \nonumber
\end{eqnarray}
Such a code allows to protect a subspace of $\kQ$ qubits against errors where
\begin{equation}
\label{eq:def_kQ}
\kQ \eqdef \dim \CoX - \dim \CoZ^\perp.
\end{equation}
$\kQ$ is called the {\em quantum dimension} of the CSS code.

\paragraph{\bf LDPC codes} 
LDPC codes are linear codes which have a sparse parity-check matrix. They can be decoded by using the {\em Tanner graph}
associated to such a parity-check matrix $\HH$. This graph is defined as follows. Assume that $\HH=(H_{ij})_{{\substack{1 \leq i \leq r \\ 1 \leq j \leq n}}}$ is an $r \times n$ matrix (where $n$ is the length of 
the code). The associated Tanner graph is bipartite and has:\\
\begin{itemize}
\item[(i)]  vertex set $V \cup C$, where
the first set $V$ is in bijection with the indices of the columns of $\HH$, say $V=\{1,\dots,n\}$ and is called the
set of {\em variable nodes}, whereas the second set $C$ is called the set of {\em check nodes} and is in bijection
with the indices of the rows of $\HH$: $C=\{\oplus_1,\dots,\oplus_r\}$. \\
\item[(ii)] edge set $E$; there is an edge between $\oplus_i$ and $j$
if and only if $H_{ij} \neq 0$ and the edge receives label $H_{ij}$ in this case.
\end{itemize}
A CSS code defined by a couple of binary code $(\CoX,\CoZ)$ is said to be a {\em quantum LDPC code} if and only if $\CoX$ and $\CoZ$ are LDPC codes.

\section{A generalization of the construction of \cite{KHIS11a}} 
\label{sec:generalization}

We show in this section how to derive for any integer $m>1$ from a pair of binary LDPC codes $(\CoX,\CoZ)$ with parity-check matrices $\hx$ and $\hz$ satisfying\\
(1) $\hx \hz^T = 0$,\\
(2) all the columns of $\hx$ and $\hz$ have exactly $2$ ``1'''s in it,\\
a pair of $2^m$-ary LDPC codes $(\CoXq,\CoZq)$ with parity-check matrices $\hxq$ and $\hzq$ satisfying\\
(1) $\hxq \hzq^T = 0$,\\
(2) all the columns of $\hx$ and $\hz$ have exactly $2$ non zero elements in it.\\
This generalizes the construction of \cite{KHIS11a} to other codes than the ones obtained from \cite{HH07a} by using the ring isomorphism
$A$ from the finite field $\fm$ to $\M{2}{m}$ which is described in Subsection II.C of \cite{KHIS11a}.

We show the existence  of the couple  $(\CoXq,\CoZq)$ by providing an efficient algorithm which outputs a couple
of matrices $(\hxq,\hzq)$ meeting  (1) and (2). To explain how the algorithm works let us bring in the following definition
\begin{definition}
To each row $k$ of $\hz$ we associate a parity-check matrix $\hx(k)$ consisting of the submatrix of $\hx$ formed by  the columns $j$ of $\hx$ such that 
$\z{k}{j} \neq 0$ and by keeping only the non zero rows in it.  Let $\GX(k)$ be the Tanner graph associated to this parity-check matrix.
\end{definition}

The crucial point is the following lemma
\begin{lemma}
\label{lem:degree}
The degree of every variable node of $\GX(k)$ is two, whereas the degree of every check node is an even positive number.
\end{lemma}

\begin{proof}
The fact that the degree of every variable node is exactly two is a direct consequence of the fact that the columns
of $\hx(k)$ are all of weight $2$ since the columns of $\hx$ have exactly this property. The second claim about the
degree of the check nodes is a consequence of $\hx \hz^T = 0$. This can be verified as follows. 
Each check node corresponds to a row of $\hx(k)$ which corresponds itself to
 some row of $\hx$. We denote such a row by  $i$. 
The degree of the check node corresponding to $i$  is nothing but the weight of row $i$ of $\hx(k)$.
It is equal to the number of $j$'s such that we both have
$\x{i}{j}=\z{k}{j}=1$. 
Notice that $\hx \hz^T = 0$ implies  in particular that
\begin{equation}
\label{eq:ortho}
\sum_{j} \x{i}{j} \z{k}{j} = 0
\end{equation}
This implies the aforementioned claim about the degree of the check node, since the aforementioned number of $j$'s
is necessarily even in order to meet \eqref{eq:ortho}.
\end{proof}

Since the degrees of all the vertices of $\GX(k)$ is even, $\GX(k)$ can be decomposed in an edge-disjoint 
subset of cycles $\SCycX{k}$. Each variable node vertex $j$ belongs to a unique cycle of this kind whereas a check node $\oplus_i$ may belong
to several cycles of $\SCycX{k}$. Our strategy to ensure that
there is a choice of $\hxq$ and $\hzq$ meeting Condition \eqref{eq:condition1} and $\hxq \hzq^T = 0$ is to look 
for solutions which satisfy for all rows $k$ of $\hzq$, 
all cycles $\Cyc$  of $\SCycX{k}$,  and all check nodes $\oplus_i$ belonging to $\Cyc$
\begin{equation}
\label{eq:preortho}
\sum_{j : \oplus_i,j \in \Edge} \xq{i}{j} \zq{k}{j} = 0
\end{equation}
where we denote by $\Edge$ the set of edges of $\Cyc$. Notice that there are exactly two variable nodes which are adjacent to
 $\oplus_i$ in $\Cyc$. The first point is that the sum $\sum_{j} \xq{i}{j} \zq{k}{j}$ can be decomposed as a
 sum $\sum_{\Cyc : \Cyc \in \SCycX{k}, \oplus_i \in \Cyc} \sum_{j : \oplus_ij \text{ edge of } \Cyc} \xq{i}{j} \zq{k}{j}$
 which implies that ensuring \eqref{eq:preortho} implies \eqref{eq:ortho}  and therefore 
 $\hxq \hzq^T = 0$. Moreover the code associated to the cyclic Tanner graph $\Cyc$ is non trivial if and only if
 the product of its labels on its cycle is equal to $1$. We define here for a Tanner graph the product over a cycle by 

\begin{defn}[{\bf product over a cycle of a Tanner graph}]
Let $\Cyc = v_1, c_1, v_2, \ldots, c_k, v_1$ be a cycle in the Tanner graph code. Then the product over this cycle is the product of all the coefficients of the edges over this cycle, with a power $1$ if it is a check-to-node edge, and $-1$ if it is node-to-check.
We denote this product by $\Pi(\Cyc)$.
\end{defn}

It is namely well known that
\begin{proposition}\label{pr:cycle}
The code associated to Tanner graph which is a unique cycle is not reduced to the zero codeword if and only if
the product of the labels over the cycle is equal to $1$. In such a case, all the non-zero codewords have only 
non-zero positions.
\end{proposition}
The proof of this proposition is given in the appendix.

The algorithm for choosing the entries of $\hxq$ and $\hzq$ is described below as Algorithm 
\ref{al:whole}.
\begin{algorithm} \caption{Choosing the entries of $\hxq$ and $\hzq$ \label{al:whole} }
\begin{algorithmic}
\STATE{Choose the entries $\xq{i}{j}$ of $\hxq$ such that for all
rows $k$ of $\hzq$ and all cycles of $\SCycX{k}$ the product of the labels $\xq{i}{j}$ along these cycles is equal to $1$.}
\FORALL{rows $k$ of $\hzq$}
\FORALL{cycles $\Cyc$ of $\SCycX{k}$}
\STATE{Choose non-zero entries $\zq{k}{j}$ for all variable nodes $j$ of  $\Cyc$ such that \eqref{eq:preortho} holds for all
edges of $\Cyc$.
}
\ENDFOR
\ENDFOR
\end{algorithmic}
\end{algorithm}
The fact that the $\zq{k}{j}$'s can be chosen to be different from zero comes from the fact 
that   the product of the labels $\x{i}{j}$ along $\Cyc$ is equal to $1$ and from
Proposition \ref{pr:cycle}. It just amounts to choose a non-zero codeword in the code whose
Tanner graph is given by $\Cyc$ and the labels of the edges are given by the $\xq{i}{j}$'s.
This leads to two matrices $\hxq$ and $\hzq$ which satisfy Condition \eqref{eq:condition1} and $\hxq \hzq^T = 0$.
Finally, it remains to explain how we choose the entries   $\xq{i}{j}$ of $\hxq$. We will actually provide an algorithm
which provides a stronger condition on the $\xq{i}{j}$'s, namely that 
\begin{equation}
\label{eq:allcycles}
\text{for all cycles $\Cyc$ of $\GX$, } \Pi(\Cyc) = 1.
\end{equation}
The fact that the product over all cycles of $\GX$ will be equal to $1$ (and not only the cycles 
of the subgraphs $\GX({k})$) will be quite useful when applied to
the toric code and this stronger condition can be met with Algorithm \ref{al:product}
which gives a very large choice for the coefficients.

\begin{algorithm}\caption{Algorithm to ensure \eqref{eq:allcycles} \label{al:product} }
\begin{algorithmic}
\FORALL{check nodes $\oplus_k$ of the Tanner graph $\GX$ associated to $\hx$}
\STATE{Choose arbitrarily a non zero element $a_k$ and non-zero elements $a_{jk}$ for 
all variable nodes $j$ adjacent to $\oplus_k$.}
\ENDFOR
\FORALL{variable nodes $j$ of  $\GX$}
\STATE{$\xq{i}{j} \leftarrow a_i b_{ij} b_{kj}$}
\STATE{$\xq{k}{j} \leftarrow a_k b_{kj} b_{ij}$}
\COMMENT{Here $\oplus_i$ and $\oplus_k$ denote the two check nodes adjacent to $j$.}
\ENDFOR
\end{algorithmic}
\end{algorithm}
\begin{proof}{(of correctness of Algorithm \ref{al:product})}
Let $\Cyc$ be a cycle of $\GX$. Let us prove that $\Pi(\Cyc)=1$. This product can be written
as 
$$\Pi(\Cyc) = \Pi_{\text{check nodes } \oplus_k \text{ in } \Cyc} f(k),$$
where 
$f(k)$ counts the contribution to the product which involves terms which depend on 
$k$. By denoting by $j$ and $l$ the two variable nodes 
adjacent to $\oplus_k$ in the cycle and by 
$\oplus_i$ and $\oplus_m$  the two other check nodes which are adjacent 
in the cycle to $j$ and $l$ respectively we 
can decompose $f(k)$ as 
$$
f(k) = g(ij)g(jk)g(kl)g(lm)
$$
where $g(ab)$ gives the part of the contribution to $\Pi(\Cyc)$ stemming from 
edge $ab$ by keeping only elements of the product which depend on $k$.
We observe now that $g(ij) = b_{kj}$, $g(jk)=a_k^{-1} b_{kj}^{-1}$, 
$g(kl)=a_k b_{kl}$ and $g(lm)=b_{kl}^{-1}$. 
This implies $f(k)=1$, which in turn implies that $\Pi(\Cyc)=1$.
\end{proof}

\noindent{\bf Remark:} One might wonder whether or not it is possible to obtain $q$-ary versions of $\hx$ and $\hz$ which
satisfy the orthogonality condition $\hxq \hzq^T = 0$ when the columns of $\hx$ and $\hz$ have weight greater than $2$.
While this can be easily done for certain structured constructions such as the one proposed in \cite{TZ09a}, it is not clear how to achieve
this in all generality. The difficulty is the following. Consider the code defined by a Tanner graph which is a subgraph of $\GX$ 
labelled by a certain choice of the $\xq{i}{j}$ and
which consists in codewords of the form $(\zq{k}{j})_{j: \z{k}{j} =1}$ satisfying \eqref{eq:ortho}.
All these codes (for $k$ ranging over all rows of $\hzq$) should be not reduced to the zero codeword. While this 
is easily achieved in the case of column weight $2$ essentially by the fact that the number of check nodes of the Tanner graphs $\GX(k)$ 
is always less than or equal to the number of variable nodes (since by Lemma \ref{lem:degree} the degree of the check nodes is greater than or equal to $2$ and
the degree of the variable nodes is constant and equal to $2$), this is not the case anymore when the column weight is higher.

\section{An application: the extended toric code}

\subsection{Definition of the toric code and its extended version}

The toric code (see \cite{BK98a} for more details) is a CSS 
code of length $2n^2$ which encodes
$2$ qubits. It is convenient to define the Tanner graphs $\GX$ and $\GZ$ of the couple $(\CoX,\CoZ)$ of binary codes of the CSS code simultaneously. Let $C_X$ and $C_Z$ be the set of variable nodes of 
$\GX$ and $\GZ$ respectively and we identify the variable node sets $V_X$ and $V_Z$ of both codes,
say $V_X=V_Z=V$.
These graphs are defined as follows:
\begin{eqnarray*}
V &= &\left\{(i,j) \in [0..2n-1] \times [0..2n-1]:i+j \text{ even}\right\}\\
C_X &= &\left\{(i,j) \in [0..2n-1] \times [0..2n-1]: i \text{ odd, } j \text{ even} \right\}\\
C_Z &=& \left\{(i,j) \in [0..2n-1] \times [0..2n-1]: i \text{ even, } j \text{ odd} \right\}
\end{eqnarray*}
A check node $(i,j)$ is connected to $4$ variable nodes $(i\pm 1,j\pm 1)$ in both graphs (where addition is performed modulo $2n$). The degree of
the variable nodes is of course $2$.

The construction, summarized on Fig~\ref{toric_small}, has the shape of a torus of length and width 
$2n$.

\begin{figure}[h]
\graphicspath{{./dessinscodetorique/}}
\includegraphics[scale=0.8]{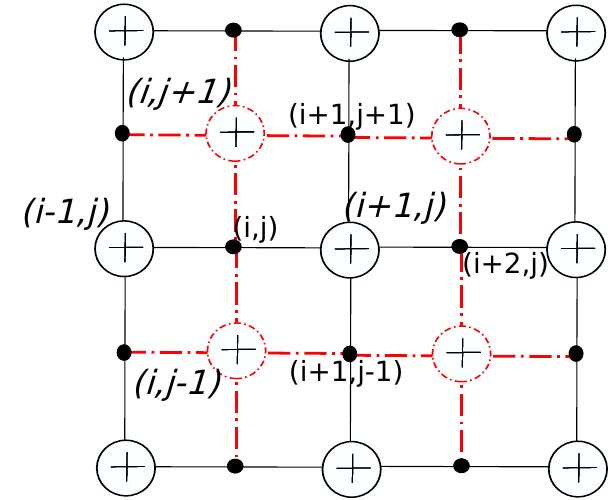}
\caption{\label{toric_small} The Tanner graph of the toric code, with both $X$ and $Z$ parts together.
The black dots represent the qubits, the dotted crosses the checks of the $Z$ part, and the black crosses the checks of the $X$ part. 
The left part is identified to the right part, and the upper part to the lower part, so that the global shape of the graph is a torus.}
\end{figure}

Even if this code has as many checks as qubits, its dimension is positive: the rank of $\hx$ and $\hz$ associated to $\GX$ and $\GZ$ is $n^2-1$ instead of $n^2$, thus the dimension is $\dim(\CoX) - \dim({\CoZ}^{\perp}) = n^2+1 - (n^2-1) = 2$ (from~\eqref{eq:def_kQ}).
The code has a rather large minimum distance \cite{Kit03a}, 
however its performances when decoded with standard belief propagation is quite bad, because of the presence of many small cycles and also 
because the (classical) minimum distance of $\CoX$ and $\CoZ$ is only $4$.

Now we construct a $q$-ary version of this code, in the same way as in Section~\ref{sec:generalization}.
In other terms, we just 
put some non-zero labels on the edges of the graph. For simplicity of 
notation we will further use $\x{i}{j}$ to design $\xq{i}{j}$, the  label in $\Fq\setminus\{0\}$ on the edge between 
check $i$ and 
node $j$.
Labeling is performed through Algorithm \ref{al:whole} by choosing the coefficients $a_i$ and $b_{jk}$ at random in 
Algorithm \ref{al:product}. We obtain a couple $(\CoXq,\CoZq)$ of $q$-ary codes 
satisfying
$$\CoZq^\perp \subset \CoXq.$$
We obtain the \emph{extended toric code} by applying the aforementioned ring isomorphism to the entries
of the parity-check matrices $\hxq$ and $\hzq$ of $\CoXq$ and $\CoZq$:
the resulting code has 
length $2mn^2$. We denote the couple of binary codes defining this toric code by $(\CoXi,\CoZi)$.

\subsection{Dimension}

Strictly speaking, by applying Algorithm \ref{al:whole}, the dimension of
$\CoXq$ minus the dimension of $\CoZq^\perp$ could be smaller than 
$\dim \CoX - \dim \CoZ$. Indeed $\hxq$ and $\hzq$ might now be of full rank and we might have
$\dim \CoXq=\dim \CoZq^\perp=n^2$. This would imply that
$\dim \CoXi = \dim \CoZi ^\perp$ and the quantum dimension of the extended toric code would be $0$.
However, when we apply Algorithm \ref{al:product} to choose the labels (so that the product of the labels 
$\x{i}{j}$ over 
{\em all} cycles of $\GX$ is equal to $1$), then it will turn out that
$$
\dim \CoXq - \dim \CoZq^\perp = \dim \CoX - \dim \CoZ^\perp = 2,
$$
so that $\dim \CoXi - \dim \CoZi=2m$. This means that

\begin{thm}[Dimension of the extended toric code]
If $\cxq$ and $\czq$ are constructed such that $\cxq$ verifies \eqref{eq:allcycles} and 
$\CoZq^\perp \subset \CoXq$, 
then the extended toric code has dimension 
$2m$.
\end{thm}

\begin{proof}
This is shown with the help of two lemmas:
\begin{lemma}
\label{l1}
If $\cxq$ verifies \eqref{eq:allcycles} 
and $\CoZq^\perp \subset \CoXq$, then $\czq$ verifies also \eqref{eq:allcycles}.
\end{lemma}

\begin{lemma}
\label{l2}
If $\cxq$ verifies \eqref{eq:allcycles}, then it has $q$-ary dimension $n^2+1$.
\end{lemma}

From these 
two lemmas, we obtain that the dimension of $\cxq$ and $\czq$ is $n^2+1$, which gives
\[
\dim \cxq - \dim {\czq}^\perp = 2
\]
This implies that the quantum dimension of the extended toric code is
\[
\dim(\CoXi) - \dim({\CoZi}^\perp) = m (\dim \cxq - \dim {\czq}^\perp) = 2m
\]
\end{proof}

The proof of the two lemmas is given in the appendix.

\subsection{Minimum distance}

Choosing the product of the labels to be equal to $1$ on all cycles of $\GX$ brings another benefit : it allows to control the minimum distance, since
we have in this case
\begin{lemma}
\label{lemma_min_distance}
$\min\{|x| \in \CoXq \setminus \CoZq^\perp\} = \min\{|x| \in \CoZq \setminus \CoXq^\perp\}=n$.
\end{lemma}
The proof  is given in the appendix. 
This implies that
\begin{thm}[minimum distance of the extended toric code]
\label{min_distance}
The minimum distance of the extended toric code is  $\geq n$.
\end{thm}

\begin{proof}
The minimum distance of the extended toric code is  the minimal weight of a word from 
$\CoXi \setminus \CoZi^\perp $ or $\CoZi \setminus \CoXi^\perp $. 
The Hamming weight of such a word is greater than or equal to the Hamming weight of the
word in $\CoXq \setminus \CoZq^\perp $ or $\CoZq \setminus \CoXq^\perp $ it corresponds to after taking the aforementioned
ring isomorphism $A$.
\end{proof}

\begin{rmk}
There is
also an upper bound on the minimum distance: it is at most 
$nm$, since a word of weight $n$ in $\Fq$ has minimal weight $n$ and maximal weight $mn$ in $\ft$.
\end{rmk}


\section{Results}
\label{sec:results}

We have implemented standard belief propagation  over $\F_{2^m}$ to decode extended toric codes for 
several values of $n$ and $m$ (see Section III of \cite{KHIS11a}) but which correspond to the same final length $2mn^2$, which is 
$1152$ here. We have 
chosen \\
(i) $m=1$, $n=24$,\\
(ii) $m=4$, $n=12$\\
(iii) $m=9$, $n=8$.\\
The channel error model is the depolarizing channel model with depolarizing probability $p$, meaning that the probability of an $X,Y$ or $Z$ error is $p/3$ which
implies that the codes $\CoXi$ and $\CoZi$ see a binary symmetric channel of probability $\frac{2p}{3}$.

The performance of belief propagation is quite bad in the binary case (that is for standard toric code), even if the qubit error rate is rather low,
the whole error is typically badly estimated. On the other hand the performances get better by moving from
$\ft$ to $\F_{16}$ and become quite good over $\F_{512}$. This is remarkable since the length of these CSS codes is constant but the rate increases with
$m$. For instance, the rate of the toric code is $\frac{1}{576}$ whereas the rate of the extended toric code over $\F_{512}$ is equal to
$\frac{1}{64}$. It would be interesting to carry over the renormalizing approach of \cite{DP10a} which improves dramatically belief propagation over standard toric codes and study how much it is able to improve the performance of standard belief propagation over these larger alphabets.
\begin{figure}[h]
\includegraphics[scale=0.55]{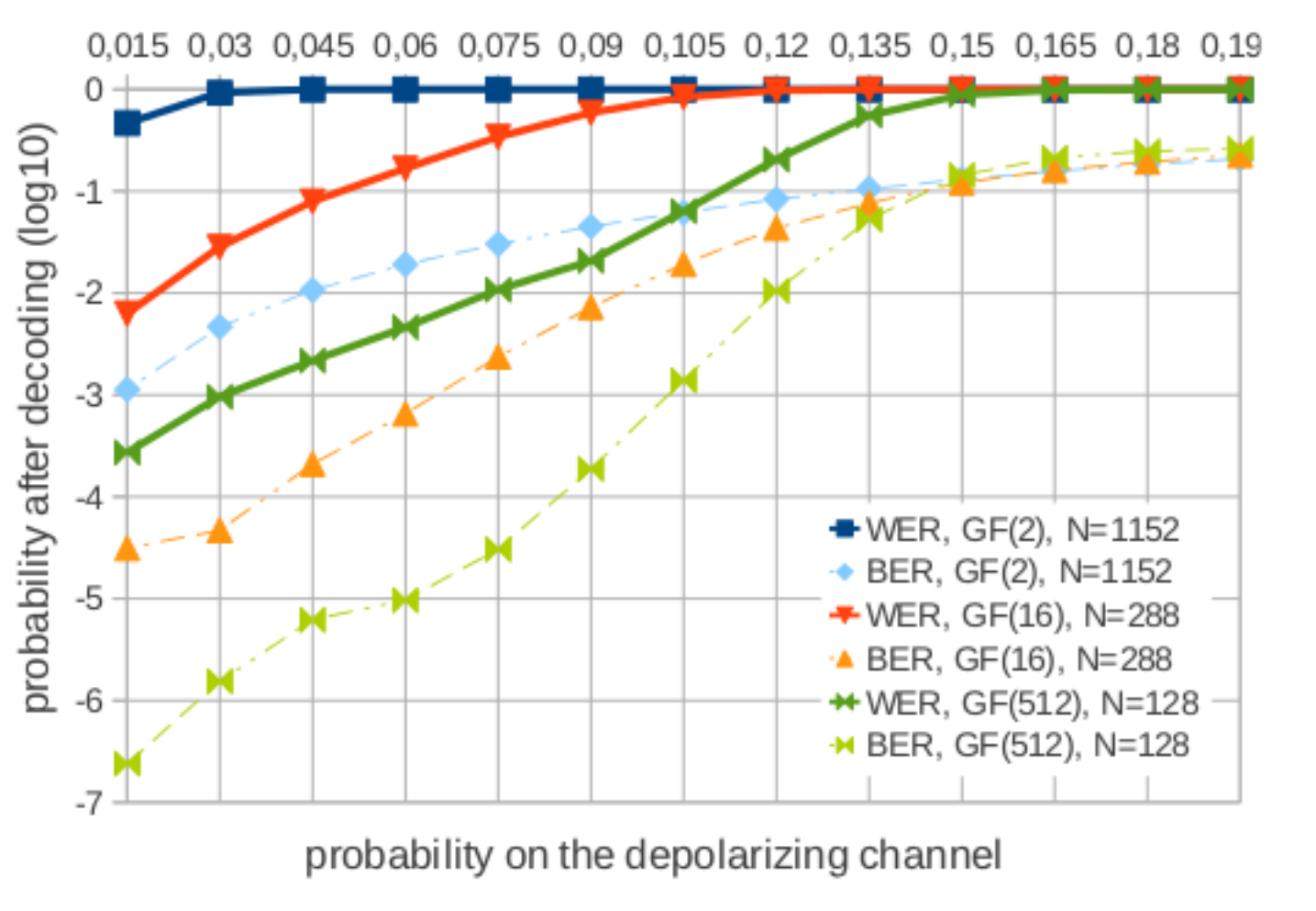}
\vspace{-0.3cm}
\caption{\label{toric_results} Word and qubit error rates for several extended toric codes.}
\end{figure}


%
%
\appendix
\begin{proof}[Proof of Proposition~\ref{pr:cycle}]
Let us consider a Tanner graph composed of a cycle $v_0, C_0, v_1, C_1, \ldots v_{k-1}, C_{k-1}$, and let $\x{i}{j}$ be the label on the edge between check $i$ and node $j$. A codeword $w_0w_1,\dots w_{k-1}$ of the code associated to this Tanner graph is such that:
\begin{eqnarray*}
\x{0}{0} w_0 & +& \x{0}{1} w_1 = 0\\
\x{1}{1} w_1 & +& \x{1}{2} w_2 = 0\\
\ldots \\
\x{k-1}{k-1} w_{k-1} & +& \x{k-1}{0} w_0 = 0
\end{eqnarray*}
This system has non-trivial solutions if and only if the determinant of this system is $0$, ie if:
\begin{eqnarray*}
\x{0}{0} \ldots \x{k-1}{k-1} + \x{0}{1} \ldots \x{k-1}{0} = 0 \\
\Longleftrightarrow \quad
\x{0}{0} \x{0}{1}^{-1} \ldots \x{k-1}{k-1} \x{k-1}{0}^{-1} = 1
\end{eqnarray*}
which means that the product over the cycle is $1$.

If this condition is verified, and one of the $w_i$'s is zero, for example $w_0$, we can see from the system that $w_1, \ldots w_{k-1}$ have to be equal to zero too. So the non-zero codewords have only non-zero positions.
\end{proof}

\begin{proof}[Proof of Lemma~\ref{l1}]

We consider here two basic types of cycles in the Tanner graphs  of $\CoXq$ and $\CoZq$:
 the minimal cycles of length $8$, and cycles of length $2n$ that go through the length or the width of the torus, we call the last ones ``big cycles''. An example is shown on Fig~\ref{big_and_small_cycles}.

\begin{figure}[h]
\begin{picture}(120,120)
\put(60,0){
\graphicspath{{./dessinscodetorique/}}
\includegraphics[scale=0.3]{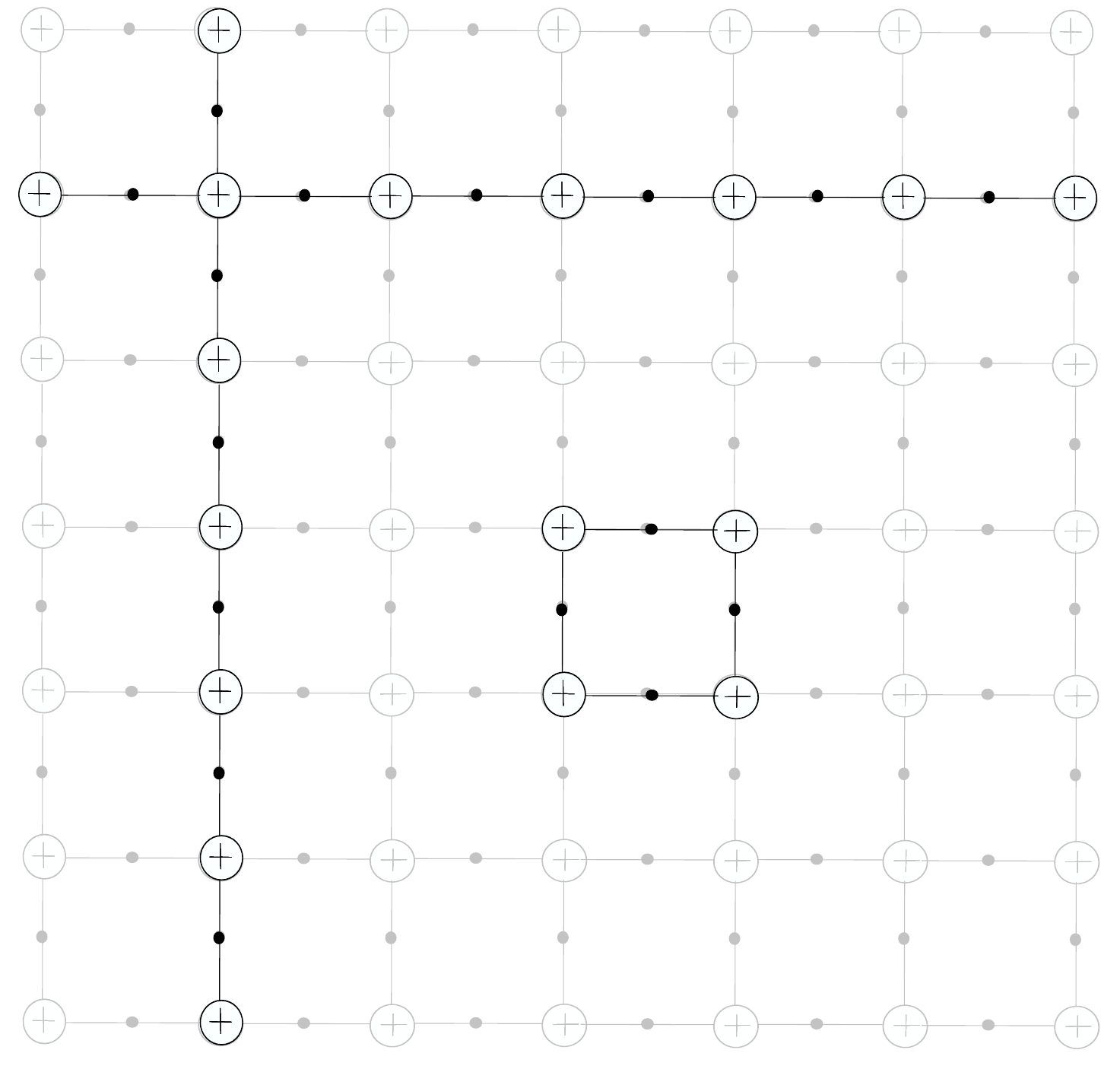}}
\end{picture}
\caption{\label{big_and_small_cycles}Two big cycles (horizontal, vertical) and a minimal cycle.}
\end{figure}

More formally,

\begin{defn}
[minimal cycle]
A \emph{minimal cycle} in the Tanner graph of $\cxq$ or $\czq$ is a cycle of the form: $(i,j), (i+1,j), (i+1, j-1), (i+1,j-2), (i, j-2), (i-1,j-2), (i-1,j-1),(i-1,j), (i,j)$ with $i+j$ even, so that $(i,j)$ is a variable node.
\end{defn}

\begin{defn}[Big cycle]
A \emph{horizontal big cycle} in the Tanner graph of $\cxq$ or $\czq$ is a cycle of the form: $(i,j), (i+1,j), \ldots (i+2n-1, j), (i,j)$, with $i+j$ even.
\footnote{recall that addition on the indices is performed modulo $2n$.} \\
A \emph{vertical big cycle} in the Tanner graph of $\cxq$ or $\czq$ is a cycle $(i,j), (i,j+1), \ldots (i, j+2n-1), (i,j)$, with $i+j$ even.
\end{defn}

Our first observation is that it is enough to prove Condition \eqref{eq:allcycles} on the minimal cycles and the big cycles of the Tanner graph of $\czq$, 
since the product of any other cycle in this Tanner graph can be decomposed as a product of products over these basic cycles.

Let us now consider a $4$-cycle which lives in the union of the two Tanner graphs of $\cxq$ and $\czq$.
It consists in  two checks (see Fig~\ref{very_small_cycle})  $A$ and $B$, that are both connected to two variable nodes $1$ and $2$. From the orthogonality constraint $\hxq \hzq^T=0$, we deduce that the labels on the edges of this cycle satisfy
\[ \x{A}{1} \z{B}{1} + \x{A}{2} \z{B}{2} = 0 \]

\begin{figure}
\begin{picture}(30,45)
\put(80,-5){
\graphicspath{{./dessinscodetorique/}}
\includegraphics[scale=1]{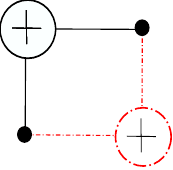}}
\put(140,0){B}
\put(134,38){1}
\put(75,40){A}
\put(82,-5){2}
\put(105,41){$\x{A}{1}$}
\put(102,-5){$\z{B}{2}$}
\put(85,15){$\x{A}{2}$}
\put(120,20){$\z{B}{1}$}
\end{picture}
\caption{\label{very_small_cycle}A $4$-cycle. The black cross $A$ is a check from $\cxq$, the dotted cross $B$ is one from $\czq$, and the dots $1$ and $2$ are the qubits where they interact}
\end{figure}
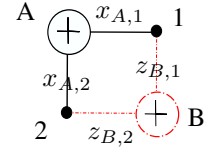

We can reformulate this:
\[
\x{A}{1} \z{B}{1} \z{B}{2}^{-1} \x{A}{2}^{-1} = 1
\]
With the following definition, we obtain in this way that the product over such  cycles of size $4$ is equal to $1$.
\begin{defn}[Product over a cycle - extended version]
The notion of \emph{product over a cycle} can be extended to the union of the Tanner graphs of  $\cxq$ and $\czq$. If $v_1, c_1, v_2, \ldots, c_k, v_1$ is a cycle in this union, the product over this cycle is the product of all the labels of the edges over this cycle, with a power: 
\begin{itemize}
\item $1$ if the edge is check-to-variable node and belongs to the $X$-part,
\item $-1$ if the edge is variable node-to-check and belongs to the $X$-part,
\item $-1$ if the edge is check-to-variable node and belongs to the $Z$-part,
\item $1$ if the edge is variable node-to-check and belongs to the $Z$-part,
\end{itemize}
\end{defn}

Now, let us look at a combination of $4$ such small cycles, as in Fig~\ref{small_cycle_Z}.

\begin{figure}[h]
\begin{picture}(80,80)
\put(60,0){
\graphicspath{{./dessinscodetorique/}}
\includegraphics[scale=1]{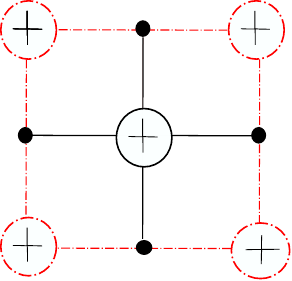}}
\put(55,75){$A$}
\put(152,75){$B$}
\put(155,0){$C$}
\put(55,0){$D$}
\put(95,30){$E$}
\put(106,80){$1$}
\put(106,-3){$3$}
\put(63,40){$4$}
\put(150,40){$2$}
\put(85,76){$\z{A}{1}$}
\put(115,76){$\z{B}{1}$}
\put(57,55){$\z{A}{4}$}
\put(145,55){$\z{B}{2}$}
\put(57,25){$\z{D}{4}$}
\put(145,25){$\z{C}{2}$}
\put(85,3){$\z{D}{3}$}
\put(115,3){$\z{C}{3}$}
\put(108,59){$\x{E}{1}$}
\put(108,22){$\x{E}{3}$}
\put(80,46){$\x{E}{4}$}
\put(119,46){$\x{E}{2}$}
\end{picture}
\caption{\label{small_cycle_Z}A view of a $X$ check ($E$) with the four related $Z$ checks around ($A, B, C, D$).}
\end{figure}
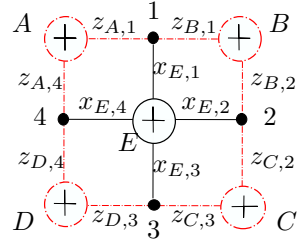

The product 
over all 
small cycles is $1$ :
\begin{eqnarray*}
\z{A}{1}^{-1} \x{E}{1}^{-1}  \x{E}{4}       \z{A}{4}      = 1 \\
\z{B}{1}      \z{B}{2}^{-1}  \x{E}{2}^{-1}  \x{E}{1}      = 1 \\
\x{E}{2}      \z{C}{2}       \z{C}{3}^{-1}  \x{E}{3}^{-1} = 1 \\
\x{E}{4}^{-1} \x{E}{3}       \z{D}{3}       \z{D}{4}^{-1} = 1
\end{eqnarray*}

By multiplying all these equations, we obtain:
\begin{eqnarray*}
\z{A}{1}^{-1} \z{A}{4} \z{B}{1} \z{B}{2}^{-1} \z{C}{2} \z{C}{3}^{-1} \z{D}{3} \z{D}{4}^{-1} = 1
\end{eqnarray*}
which is exactly the product over a minimal cycle of $\czq$.

Now, we consider another combination of $2n$  4-cycles such as in Fig~\ref{very_small_cycle}, among one direction of the torus, as shown in Fig~\ref{long_cycle_dual}. It consists, in the subgraph of both Tanner graphs, in the variable and check nodes in the cartesian product $[0..2n-1]\times\{0,1\}$. 
To simplify notation we have relabeled a variable node $(i,0)$ by $\frac{i}{2} +n$, a variable node $(i,1)$ by
$ \frac{i-1}{2}$, a check node $(i,1)$ corresponding to $\CoXq$ by $\frac{i}{2}$ and a check node
$(i,0)$ corresponding to $\CoZq$ also by $\frac{i-1}{2}$. It is summarized in Fig~\ref{long_cycle_dual}.

\begin{figure}[h]
\begin{picture}(30,60)
\put(0,0){
\graphicspath{{./dessinscodetorique/}}
\includegraphics[scale=1]{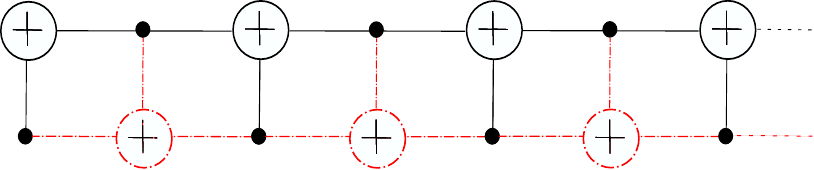}}
\put(20,5){$\z{0}{n}$}
\put(58,5){$\z{0}{n+1}$}
\put(86,5){$\z{1}{n+1}$}
\put(123,5){$\z{1}{n+2}$}
\put(153,5){$\z{2}{n+2}$}
\put(190,5){$\z{2}{n+3}$}
\put(40,25){$\z{0}{0}$}
\put(110,25){$\z{1}{1}$}
\put(178,25){$\z{2}{2}$}
\put(25,41){$\x{0}{0}$}
\put(58,41){$\x{1}{0}$}
\put(95,41){$\x{1}{1}$}
\put(124,41){$\x{2}{1}$}
\put(162,41){$\x{2}{2}$}
\put(192,41){$\x{3}{2}$}
\put(5,20){$\x{0}{n}$}
\put(65,20){$\x{1}{n+1}$}
\put(135,20){$\x{2}{n+2}$}
\put(210,20){$\x{3}{n+3}$}
\put(5,-4){$n$}
\put(70,-4){${n+1}$}
\put(140,-4){${n+2}$}
\put(210,-4){${n+3}$}
\put(42,48){$0$}
\put(112,48){$1$}
\put(182,48){$2$}
\put(40,-7){$C_{Z0}$}
\put(110,-7){$C_{Z1}$}
\put(180,-7){$C_{Z2}$}
\put(5,52){$C_{X0}$}
\put(72,52){$C_{X1}$}
\put(139,52){$C_{X2}$}
\put(210,52){$C_{X3}$}
\end{picture}
\caption{\label{long_cycle_dual} An ensemble of small cycles of $\cxq$ (black) and $\czq$ (dotted).}
\end{figure}

The product over all such cycles is $1$, ie:
\begin{equation*}
\begin{array}{l}
\x{0}{0}        \z{0}{0}      \z{0}{n}^{-1}       \x{0}{n}^{-1}   = 1 \\
\x{1}{0}^{-1}   \x{1}{n+1}    \z{0}{n+1}          \z{0}{0}^{-1}   = 1 \\
\x{1}{1}        \z{1}{1}      \z_{1}{n+1}^{-1}    \x{1}{n+1}^{-1} = 1 \\
\ldots \\
\x{n-1}{n-1}    \z{n-1}{n-1}  \z{n-1}{2n-1}^{-1}  \x{n-1}{2n-1}   = 1 \\
\x{0}{n-1}^{-1} \x{0}{n}      \z{n-1}{n}          \z{n-1}{n-1}    = 1
\end{array}
\end{equation*}

By multiplying all these 
equations, we get:
\begin{eqnarray*}
\x{0}{0}  \z{0}{n}^{-1}  \x{1}{0}^{-1} \z{0}{n+1} \ldots & \\
\x{n-1}{n-1} \z{n-1}{2n-1}^{-1} \x{0}{n-1}^{-1} \z{n-1}{n} & = 1\\
(\x{0}{0} \x{1}{0}^{-1} \ldots \x{n-1}{n-1} \x{0}{n-1}^{-1}) \times &\\
(\z{0}{n}^{-1} \z{0}{n+1} \ldots  \z{n-1}{2n-1}^{-1} \z{n-1}{n}) & = 1
\end{eqnarray*}

The first parenthesis is the product over a big cycle of $\cxq$, and the second parenthesis is the product over a big cycle of $\czq$.

It shows that if the product over a big horizontal cycle is equal to $1$ in $\cxq$, then the product over a big horizontal cycle in $\czq$ is also equal to $1$.
There is a similar proof for the vertical cycles.
\end{proof}

\begin{proof}[Proof of lemma~\ref{l2}]

First, we show that the dimension of $\cxq$ is at least $n^2+1$.

The idea is to construct a set of independent codewords associated to cycles of the Tanner graph of $\cxq$. This is obtained as 
follows.
Since all 
variable nodes of this Tanner graph 
have degree $2$, we can consider the \emph{graph of the checks}, where the vertices are the checks, and there is an edge between two vertices 
if and only if there is a variable node that is adjacent to the two checks. Informally, it just consists of the same graph where an "edge-variable node-edge" is replaced by a single edge.
We consider a spanning tree of this graph. An example of such spanning tree is shown in Fig~\ref{toric_covering}.

\begin{figure}
\includegraphics[scale=0.4]{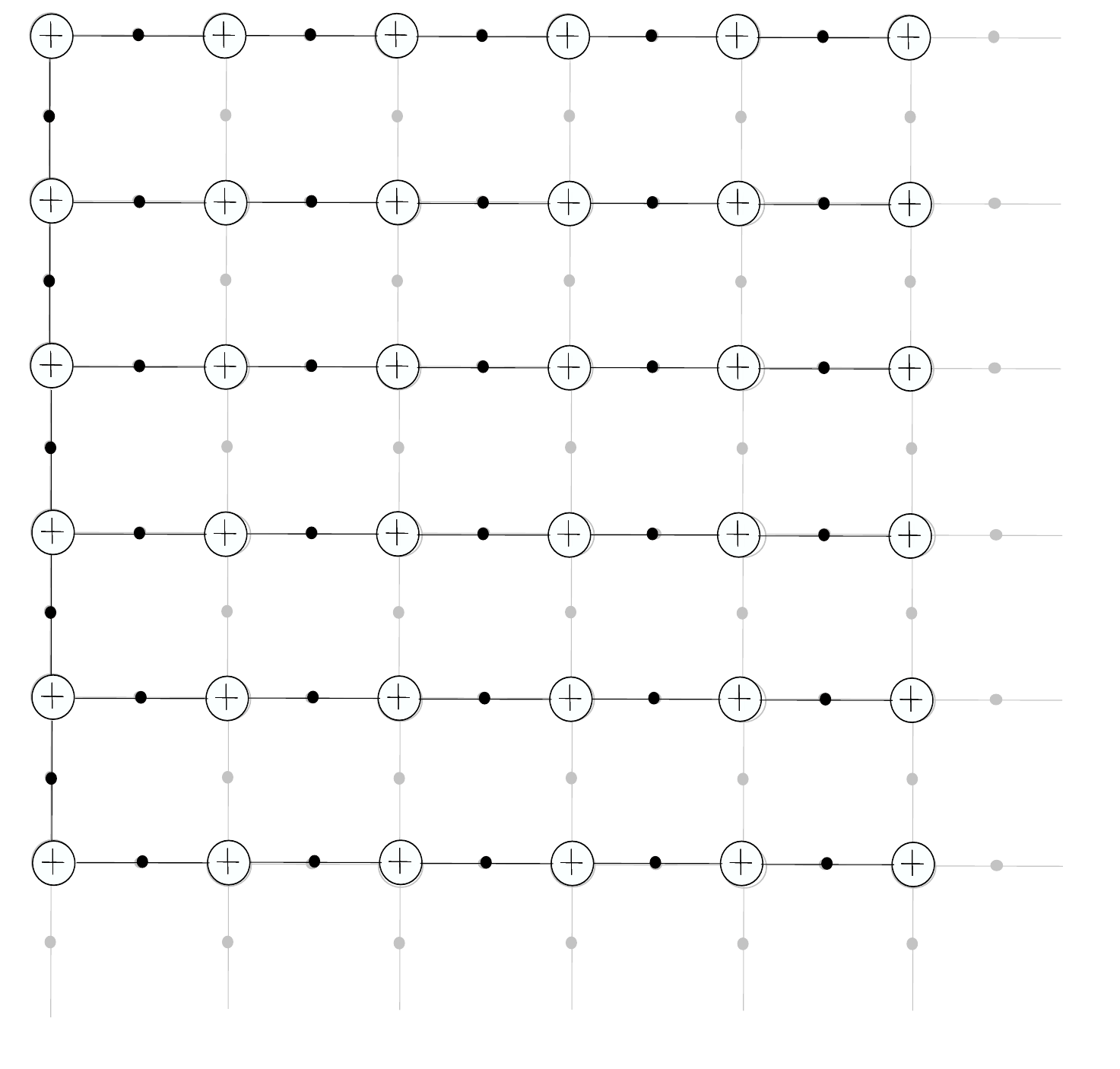}
\caption{\label{toric_covering} A spanning tree (black) of the graph of the checks. The Tanner graph of the toric code is shown in grey.}
\end{figure}

This spanning tree has of course $n^2$ checks, and therefore $n^2-1$ edges between these checks. There are $n^2+1$ other edges: let $e_1, \ldots e_{n^2+1}$be such edges. For all $i$, adding $e_i$ to the spanning tree provides a 
unique cycle, $c_i$. Let $c'_i$ be the corresponding cycle in the original Tanner graph.  
Now,  the product over each such cycle is $1$. 
From Proposition~\ref{pr:cycle}, each of these cycles provides a codeword of $\cxq$. 
These  $n^2+1$  codewords are necessarily independent, since for all the positions which correspond
to the edges $e_1, \ldots e_{n^2+1}$, exactly one of these codewords has a non zero entry (for the edge
$e_i$ it is precisely $c'_i$ which has a non zero entry for this position).

To show that this dimension is at most $n^2+1$, we remove a certain check, say check $c_r$.
We want to show that the remaining $n^2-1$ checks are independent. To obtain this, we 
prove that for any syndrome, we can construct an error that gives this syndrome.
In particular, we show that for every check $c_0$, we can get the syndrome $(0, \ldots 0,1, 0, \ldots 0)$ with $1$ at position $c_0$. 

Let $c_0, v_1, c_1, v_2, \ldots v_k, c_r$ be some path 
in the Tanner graph that links $c_0$ to $c_r$. An example of such path is shown in Fig~\ref{path_syndrome}.

\begin{figure}[h]
\begin{picture}(70,70)
\put(20,0){
\graphicspath{{./dessinscodetorique/}}
\includegraphics[scale=0.8]{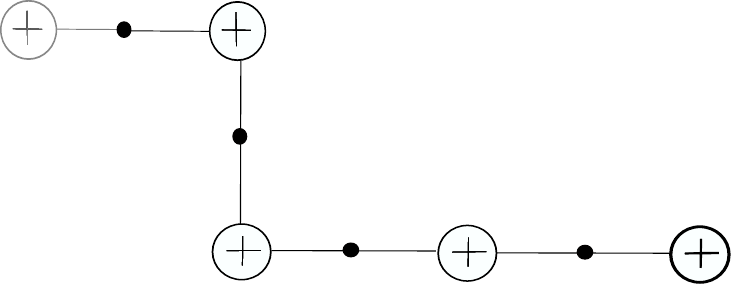}}
\put(23,50){$c_r$}
\put(185,15){$c_0$}
\put(155,12){$v_1$}
\put(130,15){$c_1$}
\put(105,12){$v_2$}
\put(67,5){$c_2$}
\put(70,32){$v_3$}
\put(169,2){$\x{0}{1}$}
\put(144,2){$\x{1}{1}$}
\put(112,2){$\x{1}{2}$}
\put(89,2){$\x{2}{2}$}
\end{picture}
\caption{\label{path_syndrome}A path between the removed check $c_r$ (grey) and some check $c_0$ (bold).}
\end{figure}
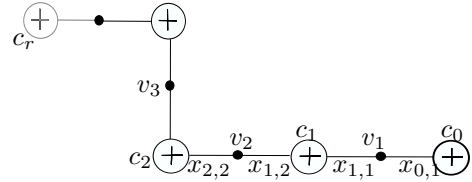

Now we construct an error $E$ that has $0$ in every position except the $v_i$'s:
\begin{itemize}
\item $E_{v_1}$ is such that the syndrome in $c_0$ is $1$, ie $E_{v_1} = 1/\x{0}{1}$
\item $E_{v_2}$ such that $c_1$ has syndrome $0$, ie $E_{v_2} =  E_{v_1}  \frac{\x{1}{1}}{\x{1}{2}}$,
and so on and so forth.
\end{itemize}

Since $c_r$ has been removed, all the checks except 
$c_0$ are satisfied.
\end{proof}


\begin{proof}[Proof of Lemma~\ref{lemma_min_distance}]

Consider an element $E$ of minimal weight in the set $\czq / \cxq^\perp$. We are going to prove that its weight is greater than or equal to
$n$. A similar proof shows that this is also the case for the minimal weight elements of $\cxq / \czq^\perp$ and this proves the lemma.

From Lemma~\ref{l2}, we know that the dimension of $\czq$ is $n^2+1$ and the dimension of $\cxq^\perp$ is $n^2-1$. Then the quotient 
$\czq/\cxq^\perp$ has dimension $2$, consequently we just need to find two independent codewords $\bar{X_1}$ and $\bar{X_2} \in \czq / \cxq^\perp$, and any such $E$ can be written as $E = E_s + \alpha_1\bar{X_1} + \alpha_2\bar{X_2}$, with $\alpha_1, \alpha_2 \in \Fq, E_s \in \cxq^\perp$, and at least one of either $\alpha_1$ or $\alpha_2$ should be non zero.

We claim that we can choose  $\bar{X_1}$ to be a codeword provided by a big vertical cycle of the Tanner graph of $\czq$ (obtained from Proposition ~\ref{pr:cycle}), and $\bar{X_2}$ being defined similarly with a big horizontal cycle. We also define $\bar{Z_1}$ and $\bar{Z_2}$, provided by respectively a big horizontal cycle and a big vertical cycle of $\cxq$.

We notice that the following inner product is non zero
\begin{eqnarray*}
\bar{X_1} \bar{Z_1}^{T} \not= 0
\end{eqnarray*}
since there is only one coordinate where both $\bar{X_1}$ and $\bar{Z_1}$ are not zero.

Note that $\bar{X_1}$ belongs to $\czq$ by definition and $\bar{X_1} \not\in \cxq^\perp$, otherwise $\bar{X_1}$ would have been orthogonal to all words of $\cxq$, including $\bar{Z_1}$. We  have the same result for $\bar{X_2}$ (and $\bar{Z_2}$). We finally just need to prove that they are independent:

Assume that $\bar{X_1} = \alpha \bar{X_2} + E_s$, with  $\alpha$ in $\Fq \setminus \{0\}$ and $E_s$ in  $\cxq$. Then we would have
\begin{eqnarray*}
\bar{X_1} \bar{Z_1}^{T} & = & \alpha \bar{X_2}\bar{Z_1}^{T} + E_s \bar{Z_1}^{T} \\
& = & \alpha \bar{X_2} \bar{Z_1}^{T}
\end{eqnarray*}
since $E_s \bar{Z_1}^{T} = 0$ because $\bar{Z_1} \in \cxq$ and $E_s \in \cxq^\perp$.

The left part is non zero, and the right part is zero, since the supports of $\bar{X_2}$ and $\bar{Z_1}$ are  disjoint. This leads to a contradiction.

Now assume $E$ is of the form $E = \alpha_1 \bar{X_1} + \alpha_2 \bar{X_2} + E_s$ with either $\alpha_1$ or $\alpha_2$ being different from
$0$.
We want to show that this error is of weight at least $n$. Assume now that $\alpha_1 \neq 0$.

We introduce $n$ shifts of $\bar{Z_1}$: $\bar{Z_1}^i$, for all $i$ even, $i \in \{0, \ldots 2n-2 \}$ which is the codeword provided by the cycle: $(0,i), (1,i), \ldots (2n-1,i), (0,i)$. They are just horizontal cycles, at different ``heights'', as shown in Fig~\ref{toric_X_1_Z_1}.

We have
\[
E_s (\bar{Z_1}^i)^T = 0
\]
since $\bar{Z_1}^i \in \cxq$ and $E_s \in \cxq^\perp$, and
\begin{eqnarray*}
\bar{X_1} (\bar{Z_1}^i)^{T} \not= 0
\end{eqnarray*}
Notice that we have for all $i$:
\begin{eqnarray*}
E (\bar{Z_1}^i)^T& = &\alpha_1 \bar{X_1} (\bar{Z_1}^i)^T + \alpha_2 \bar{X_2} (\bar{Z_1}^i)^T + E_s (\bar{Z_1}^i)^T \\
& = & \alpha_1 \bar{X_1} (\bar{Z_1}^i)^T \\
& \neq &0
\end{eqnarray*}
This implies that for all $i$, $E$ has at least a non-zero coordinate on the support of $\bar{Z_1}^i$. Since all $\bar{Z_1}^i$'s have disjoint support, it shows that $E$ has at least $n$ non-zero coordinates.
A similar reasoning holds in the case $\alpha_2 \neq 0$ by multiplying by $\bar{Z_2}$ this time.

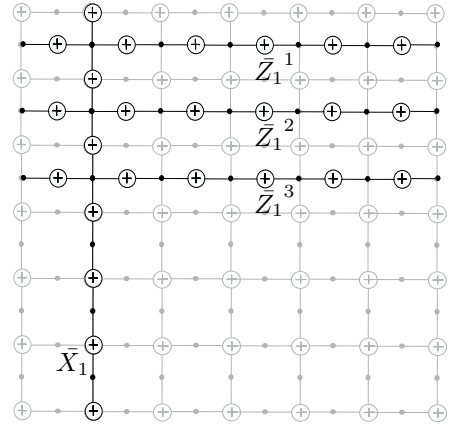
\begin{figure}[h]
\begin{picture}(160,160)
\put(20,0){
\graphicspath{{./dessinscodetorique/}}
\includegraphics[scale=0.4]{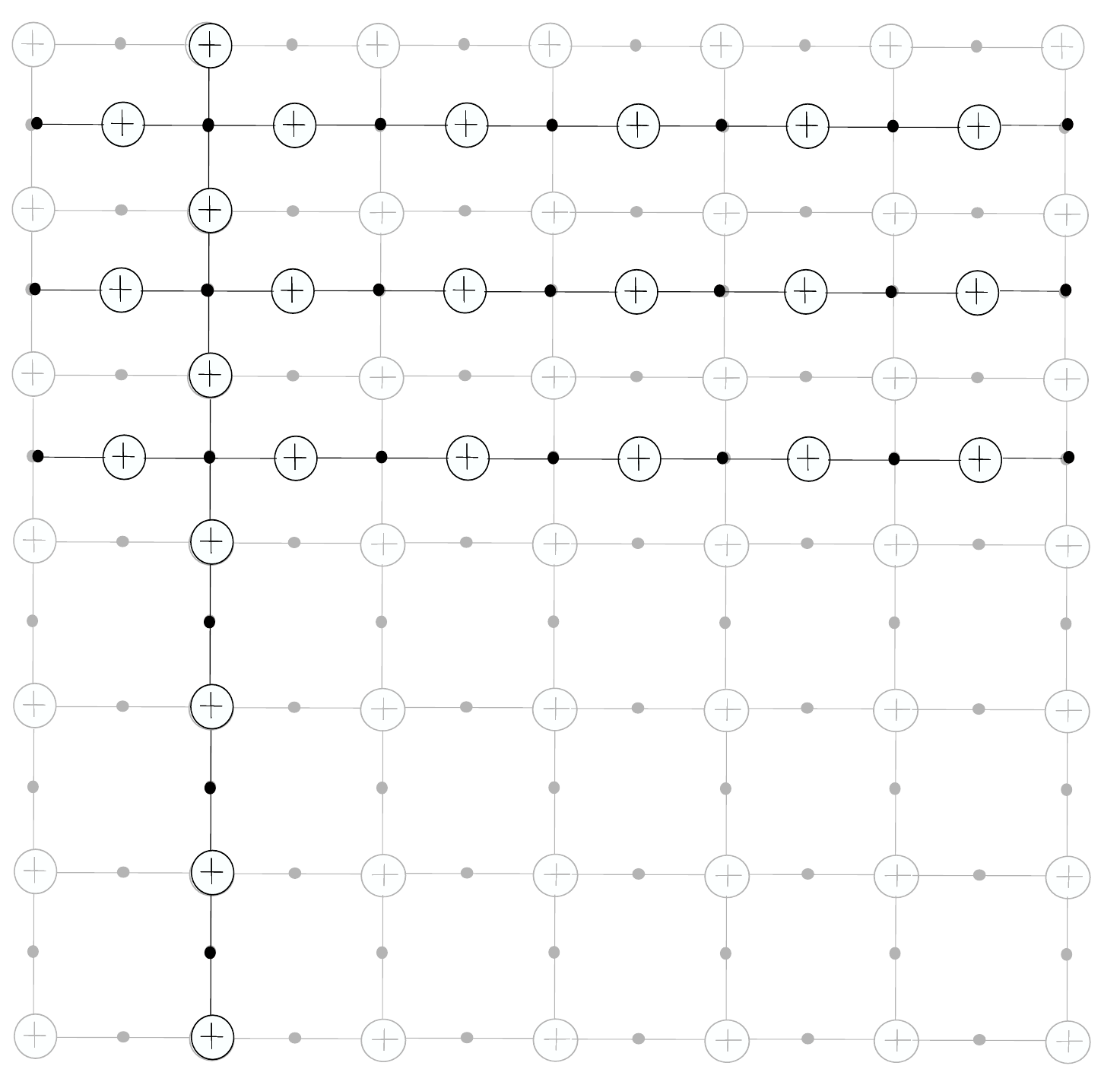}}
\put(45,20){$\bar{X_1}$}
\put(120,130){$\bar{Z_1}^1$}
\put(120,105){$\bar{Z_1}^2$}
\put(120,80){$\bar{Z_1}^3$}
\end{picture}
\caption{\label{toric_X_1_Z_1}The toric code (only the $Z$-part is displayed) of length $6$ with $\bar{X_1}$ and several variants of $\bar{Z_1}$}
\end{figure}

\end{proof}

\bibliography{qubib,code}

\begin{thebibliography}{MMM04}

\bibitem[Aly07]{Aly07b}
S.~A. Aly.
\newblock A class of quantum {LDPC} codes derived from {L}atin squares and
  combinatorial objects.
\newblock Technical report, Department of Computer Science, Texas {A{\&}M}
  University, April 2007.

\bibitem[Aly08]{Aly08a}
S.~A. Aly.
\newblock A class of quantum {LDPC} codes constructed from finite geometries.
\newblock In {\em Proceedings of IEEE GLOBECOM 2008}, pages 1--5, December
  2008.

\bibitem[BK98]{BK98a}
S.~B. Bravyi and A.~Kitaev.
\newblock Quantum codes on a lattice with boundary, 1998.
\newblock quant-ph/9811052.

\bibitem[COT05]{COT05a}
T.~Camara, H.~Ollivier, and J.-P. Tillich.
\newblock Constructions and performance of classes of quantum {LDPC} codes,
  2005.
\newblock {\tt arXiv:quant-ph/0502086v2}.

\bibitem[COT07]{COT07a}
T.~Camara, H.~Ollivier, and J.-P. Tillich.
\newblock A class of quantum {LDPC} codes: construction and performances under
  iterative decoding.
\newblock In {\em Proceedings of ISIT 2007}, pages 811--815, Nice, June 2007.
  IEEE.

\bibitem[CRSS98]{CRSS98a}
A.~R. Calderbank, E.~M. Rains, P.~W. Shor, and N.~J.~A. Sloane.
\newblock Quantum error correction via codes over {GF(4)}.
\newblock {\em IEEE Trans. Info. Theor.}, 44:1369, 1998.

\bibitem[CS96]{CS96a}
A.~R. Calderbank and P.~W. Shor.
\newblock Good quantum error-correcting codes exist.
\newblock {\em Phys. Rev. A}, 54:1098--1105, 1996.

\bibitem[DCP10]{DP10a}
G.~Duclos-Cianci and D.~Poulin.
\newblock Fast decoders for topological quantum codes.
\newblock {\em Phys. Rev. Lett.}, 104(050504), 2010.

\bibitem[Djo08]{Djo08a}
I.~B. Djordjevic.
\newblock Quantum {LDPC} codes from incomplete block designs.
\newblock {\em IEEE Communication Letters}, 12(5):389--391, May 2008.

\bibitem[Gal63]{Gal63a}
R.~G. Gallager.
\newblock {\em Low Density Parity Check Codes}.
\newblock M.I.T. Press, Cambridge, Massachusetts, 1963.

\bibitem[GFL08]{LG08a}
J.~Garcia-Frias and K.~Liu.
\newblock Design of near-optimum quantum error-correcting codes based on
  generator and parity-check matrices of {LDGM} codes.
\newblock In {\em Proceedings of CISS}, pages 562--567, Princeton, March 2008.

\bibitem[Got97]{Got97a}
D.~Gottesman.
\newblock {\em Stabilizer codes and quantum error correction}.
\newblock PhD thesis, California Institute of Technology, Pasadena, CA, 1997.

\bibitem[HB68]{HakBre68}
S.~L. Hakimi and J.~G. Bredeson.
\newblock Graph theoretic error-correcting codes.
\newblock {\em IEEE Trans. on Information Theory}, 14:584--591, 1968.

\bibitem[HBD08]{HBD08a}
M-H. Hsieh, T.~A. Brun, and I.~Devetak.
\newblock Quantum quasi-cyclic low-density parity check codes, March 2008.
\newblock { \tt arXiv:0803.0100v1 [quant-ph]}.

\bibitem[HEA05]{HEA05}
X.~Hu, E.~Eleftheriou, and D.M. Arnold.
\newblock Regular and irregular progressive edge-growth {T}anner graphs.
\newblock {\em IEEE Trans. on Information Theory}, 51(1):386--398, January
  2005.

\bibitem[HI07]{HH07a}
M.~Hagiwara and H.~Imai.
\newblock Quantum quasi-cyclic {LDPC} codes.
\newblock In {\em Proc. IEEE Int. Symp. Info. Theo. (ISIT'07)}, pages 806--811,
  Nice, June 2007. IEEE.

\bibitem[Hu02]{Hu02}
X.~Hu.
\newblock {\em Low-delay low-complexity error-correcting codes on sparse
  graphs}.
\newblock PhD thesis, EPFL, {L}ausanne, {S}witzerland, 2002.

\bibitem[IM07]{IofMez07a}
L.~Ioffe and M.~M{\'e}zard.
\newblock Asymmetric quantum error-correcting codes.
\newblock {\em Phys. Rev. Lett. A.}, 2007.

\bibitem[KHIK11]{KHIS11a}
K.~Kasai, M.~Hagiwara, H.~Imai, and Sakaniwa K.
\newblock Quantum error correction beyond the bounded distance decoding limit.
\newblock {\em IEEE Trans. Info. Theor.}, 2011.
\newblock to appear, see also arXiv:1007.17782v2[cs.IT].

\bibitem[Kit03]{Kit03a}
A.~Y. Kitaev.
\newblock Fault-tolerant quantum computation by anyons.
\newblock {\em Ann. Phys.}, 303:2, 2003.

\bibitem[LGF06]{LG06a}
H.~Lou and J.~Garcia-Frias.
\newblock On the application of error-correcting codes with low-density
  generator matrix over different quantum channels.
\newblock In {\em Proceedings of Turbo-coding 2006}, Munich, April 2006.

\bibitem[MMM04]{MMM04a}
D.~J.~C. MacKay, G.~Mitchison, and P.~L. MacFadden.
\newblock Sparse graph codes for quantum error-correction.
\newblock {\em IEEE Trans. Info. Theor.}, 50(10):2315--2330, 2004.

\bibitem[Pos01]{Pos01a}
M.~S. Postol.
\newblock A proposed quantum low density parity check code, 2001.
\newblock available at {{\tt arXiv:quant-ph/0108131v1}}.

\bibitem[SMK08]{SRK08a}
K.~P. Sarvepalli, M.R{\"o}tteler, and A.~Klappenecker.
\newblock Asymmetric quantum {LDPC} codes.
\newblock In IEEE, editor, {\em Proceedings of ISIT 2008}, pages 305--309,
  Toronto, Canada, July 2008.

\bibitem[Ste96]{Ste96b}
A.~M. Steane.
\newblock Multiple particle interference and quantum error correction.
\newblock {\em Proc. R. Soc. Lond. A}, 452:2551--2577, 1996.

\bibitem[TL10]{TL10a}
P.~Tan and J.~Li.
\newblock Efficient quantum stabilizer codes: {LDPC} and {LDPC}-convolutional
  constructions.
\newblock {\em IEEE Trans. Info. Theor.}, 56(1):476--491, 2010.

\bibitem[TZ09]{TZ09a}
J.-P. Tillich and G.~Z{\'e}mor.
\newblock Quantum {LDPC} codes with positive rate and minimum distance
  proportional to {$n^{\frac{1}{2}}$}.
\newblock In {\em Proceedings of ISIT 2009}, pages 799--803, July 2009.

\end{thebibliography}
\bibliographystyle{alpha}

\end{document}